\documentclass[conference]{IEEEtran}
\IEEEoverridecommandlockouts
\usepackage{cite}
\ifCLASSINFOpdf
   \usepackage[pdftex]{graphicx}
\else
\fi

\usepackage{amsmath,amssymb,amsfonts}
\usepackage{amsthm}
\usepackage[]{algorithm2e}
\usepackage{graphicx}
\usepackage{textcomp}
\usepackage{xcolor}
\usepackage{pgfplots}
\usepackage{tikz}
\usepackage{framed}
\usepackage{algpseudocode}
\usepackage{balance}

\pgfplotsset{grid style={gray!50}}
\pgfplotsset{minor grid style={gray!50}}

\newtheorem{theorem}{Theorem}
\newtheorem{proposition}{Proposition}

\newtheorem{alg}{Algorithm}

\begin{document}
\bstctlcite{IEEEexample:BSTcontrol}

\title{Training Channel Selection for Learning-based 1-bit Precoding in Massive MU-MIMO\thanks{\textcopyright\ 2020 IEEE. This work has been accepted for publication in 2020 IEEE International Conference on Communications Workshops (ICC Workshops). The final published version is available at DOI: 10.1109/ICCWorkshops49005.2020.9145443.}\\
}

\author{\IEEEauthorblockN{Sitian Li\IEEEauthorrefmark{1}, Andreas Burg\IEEEauthorrefmark{1}, and Alexios Balatsoukas-Stimming\IEEEauthorrefmark{2}}
\IEEEauthorblockA{\IEEEauthorrefmark{1}Telecommunication Circuits Laboratory, \'{E}cole polytechnique f\'{e}d\'{e}rale de Lausanne, Switzerland\\
\IEEEauthorrefmark{2}Department of Electrical Engineering, Eindhoven University of Technology, The Netherlands}%
}

\maketitle

\begin{abstract}
Learning-based algorithms have gained great popularity in communications since they often outperform even carefully engineered solutions by learning from training samples. In this paper, we show that the selection of appropriate training examples can be important for the performance of such learning-based algorithms. In particular, we consider non-linear 1-bit precoding for massive multi-user MIMO systems using the C2PO algorithm. While previous works have already shown the advantages of learning critical coefficients of this algorithm, we demonstrate that straightforward selection of training samples that follow the channel model distribution does not necessarily lead to the best result. Instead, we provide a strategy to generate training data based on the specific properties of the algorithm, which significantly improves its error floor performance.
\end{abstract}


\section{Introduction}
Massive multiple-input multiple-output (MIMO) systems have numerous benefits~\cite{larsson_massive_2014}, such as focusing energy into ever-smaller regions of space which improves throughput and radiated energy efficiency. 
However, they also require hundreds of antennas at each base station which cancels out some of these benefits and often results in high system costs and power consumption when compared to, e.g., 4G LTE base stations. 

Reducing the resolution of the digital-to-analog converters (DACs) in each antenna is one way to alleviate this problem. Specifically, the work of~\cite{jacobsson_quantized_2017} limited the DAC resolution to keep the power budget within tolerable levels and showed that the distortion caused by the 1-bit DACs averages out when many transmit antennas are available. Moreover, due to the large number of antennas, the high dimensionality of the transmitted signal can be used to compensate for the low DAC resolution when appropriate signal pre-processing schemes are used, such as the precoding in~\cite{jacobsson_quantized_2017}. The work of~\cite{castaneda_1-bit_2017} examined the extreme case of 1-bit precoding and proposed the two following solutions to derive the 1-bit quantized transmitted signal vectors: biConvex 1-bit PrecOding (C1PO) and its low-complexity variant called C2PO. The work of \cite{balatsoukas-stimming_neural-network_2019} developed the neural network optimized C2PO (NNO-C2PO) algorithm, which delivers significantly improved performance over C2PO with automated learning-based parameter tuning. Specifically, the NNO-C2PO algorithm learns the tunable parameters in C2PO from large training sets by unfolding the C2PO iterations. We note that the concept of unfolding an iterative algorithm has been widely used in different applications, such as compressive sensing, sparse coding \cite{borgerding_onsager-corrected_2016}, MIMO detection \cite{samuel_learning_2019}, and several other communications-related applications~\cite{balatsoukas-stimming_deep_2019}.

The C2PO precoder described in \cite{castaneda_1-bit_2017} generates binary DAC inputs based on the symbol vectors to be transmitted and on channel state information. However, its performance is highly dependent on a number of parameters~\cite{balatsoukas-stimming_neural-network_2019} 
which ideally should be chosen as a function of the channel, while the precoder is generally used on rapidly time-varying channels in practice. Thus, in the case of one-time tunable parameters, these should be learned properly from a carefully crafted training set, such that the trained result performs well for a wide range of channels.

\subsubsection*{Contributions}
In this paper, we propose some guidelines for generating training channels for learning C2PO parameters as defined in \cite{castaneda_1-bit_2017}, based on the unfolded network structure in \cite{balatsoukas-stimming_neural-network_2019}. We first suggest evaluating the performance of the C2PO 1-bit precoding algorithm in the error floor region, which essentially ignores the influence of the noise and is more convenient for comparison purposes. Based on this metric, we show that the performance of NNO-C2PO can be significantly improved with respect to~\cite{balatsoukas-stimming_neural-network_2019} when selecting the training channel set more carefully. To this end, we propose a simple, but effective channel selection scheme for parameter training in C2PO, only based on the 2-norm of the channel matrix, which is linked to the overall performance of C2PO. We note that, while we focus on the C2PO algortihm, this training set selection method can also be used for other iterative algorithms (e.g., massive MIMO detection~\cite{samuel_learning_2019}) that have a similar structure.

\subsubsection*{Outline}
The remainder of this paper is organized as follows. In Section~\ref{sec:system_model}, we explain 1-bit precoding for the MU-MIMO downlink using the C2PO algorithm and we indicate the parameters that can be learned from a given set of training channels using an unfolded C2PO structure. In Section~\ref{sec:channel_selection}, we provide guidelines for training channel selection for learning the parameters of NNO-C2PO that is based on the 2-norm of the candidate channels. In Section~\ref{sec:results} we show simulations results and we perform a comparison with other training channel selection schemes. Finally, we conclude the paper in Section~\ref{sec:conclusion}.

\section{System Model and 1-bit Precoding using C2PO}
\label{sec:system_model}
The downlink system described in \cite{jacobsson_quantized_2017} consists of $U$ users and $B$ base station antennas, where $U \leq B$. The narrowband downlink channel can be modeled as $\mathbf{y} = \mathbf{H}\mathbf{x} + \mathbf{n}$, where $\mathbf{H} \in~\mathbb {C}^{U \times B}$ is the channel matrix, $\mathbf{x} \in \mathbb{C}^B$ is the transmitted vector, and $\mathbf{y} \in \mathbb{C}^U$ is the received vector. The precoding procedure uses the symbols to be transmitted, denoted by $\mathbf{s}$, and assumes knowledge of channel matrix $\mathbf{H}$ to construct the vector $\mathbf{x} \in \mathbb{C}^U$ to transmit, such that $\mathbf{y}$ on the receiver side is demodulated as close to $\mathbf{s}$ as possible. This setting results in a reconstruction problem
$\mathop{\mathrm {arg\;min}}_{\mathbf {x}, \alpha} \|\alpha\mathbf{s} - \mathbf{H}\mathbf{x}\|^2$, where $\alpha$ is related to the precoding factor as described in \cite{jacobsson_quantized_2017}. Given $\mathbf{x}$, the optimal value of $\alpha$  is $\hat{\alpha} = \mathbf{s^H}\mathbf{H}\mathbf{x}/\|\mathbf{s}\|_2^2$, which leads to the following objective function 

\begin{align}
  \quad \hat { {\mathbf {x}}} 
= \mathop {\mathrm {arg\;min}} _{ {\mathbf {x}}\in {\mathbb {C}}^{B}} 
\,\, \mathopen {}\left \lVert{ {\mathbf {A}} { {\mathbf {x}}} }\right \rVert ^{2}_{2}, 
\label{eq:recon_general}
\end{align}
where $\mathbf{A} = (\mathbf{I}_U - \mathbf{s}\mathbf{s}^H/\|\mathbf{s}\|_2^2)\mathbf{H}$.

% here we talked about the objective problem.

In 1-bit massive MU-MIMO systems, the in-phase and quadrature components are generated separately using a pair of 1-bit DACs and hence, the per-antenna quaternary transmit alphabet is $\mathcal {X} = \{\pm l, \pm jl \}$ for a given $ l>0 $ that determines the transmit power \cite{castaneda_1-bit_2017}. We further assume that the precoded vector satisfies an instantaneous power constraint $\|x\|^2 \leq P$ so that $l$ can be expressed as $l = \sqrt{P/(2B)}$. This fact limits the feasible region of the objective function of (\ref{eq:recon_general}) into ${\mathbf {x}}\in {\mathcal {X}}^{B}$
\begin{align}  \quad \hat { {\mathbf {x}}} 
= \mathop {\mathrm {arg\;min}} _{ {\mathbf {x}}\in {\mathcal {X}}^{B}} 
\,\, \mathopen {}\left \lVert{ {\mathbf {A}} { {\mathbf {x}}} }\right \rVert ^{2}_{2}.
\end{align} 
If $B$ is much smaller than $U$, this reconstruction problem has significantly more variables than constraints. 

By replacing the feasible region $\mathcal {X}^{B}$ by its convex hull, denoted by $\mathcal {B}^B$, and adding a regularizer $-\frac{\delta}{2}\|\mathbf{x}\|_2^2$ to avoid the trivial all-zero solution, the objective function of C2PO is 
\begin{align}
\hat{\mathbf{x}} = \mathop {\mathrm {arg\;min}} _{ {\mathbf {x}}\in {\mathcal {B}}^{B}} 
\frac{1}{2}\|\mathbf{A}\mathbf{x}\|_2^2 - \frac{\delta}{2}\|\mathbf{x}\|_2^2,
\label{equ:obj_C2PO}
\end{align}
which can be further written as 
\begin{align}
\hat{\mathbf{x}} = \mathop {\mathrm {arg\;min}} _{ {\mathbf {x}}\in {\mathbb {C}}^{B}}  f(\mathbf{x}) + g(\mathbf{x}),
\label{equ:obj_C2PO_f_g}
\end{align}
where $f(\mathbf{x}) = \frac{1}{2}\|\mathbf{A}\mathbf{x}\|_2^2$, $g(\mathbf{x})  = \chi(\mathbf{x}\in\mathcal {B}^B) - \frac{\delta}{2} \|\mathbf{x}\|_2^2$, and $\chi$ is a \textit{chracteristic function} that is zero if the condition $\mathbf{x} \in \mathcal{B}^B$ is met and infinity otherwise. 

C2PO solves \eqref{equ:obj_C2PO} iteratively, given a fixed step size $\tau$, using the following algorithm:

\begin{oframed}

\vspace{-0.3cm}

\begin{alg}[C2PO \cite{castaneda_1-bit_2017}] \label{alg:C2PO} 

Initialize \mbox{$\mathbf{x}^{(0)}=\mathbf{H}^H\mathbf{s}$}. Fix $\tau$ and $\rho$. For every iteration $t=1,2,\ldots,t_\text{max}$, compute: 
\begin{align}
\mathbf{x}^{(t+\frac{1}{2})} &= \mathbf{x}^{(t)} - \tau \mathbf{A}^H\mathbf{A} \mathbf{x}^{(t)}  \label{eq:c2postep1} \\
\mathbf{x}^{(t+1)} & = \mathrm{prox} ( {\mathbf{x}^{(t+\frac{1}{2})}};\rho,\xi), \label{eq:c2postep2}
\end{align}
where $\rho = \frac{1}{1-\tau\delta}$ and $\xi = \sqrt{\frac{P}{2B}}$.  Finally, quantize the output $\mathbf{x}^{(t_\text{max})}$ to the set $\mathcal{X}^B$.
\end{alg}
\vspace{-0.2cm}
\end{oframed}
\noindent Eq. \eqref{eq:c2postep1} is an update of $\mathbf{x}$, towards the gradient of $f(\mathbf{x})$ and  (\ref{eq:c2postep2}) is the following proximal operation of $g(\mathbf{x})$  
\begin{align}
\mathrm{prox}_g (\mathbf{x};\rho,\xi)\!=\!\mathrm{clip}\!\left(\!\rho\Re\{\mathbf{x}\}, \xi\right) + j~\text{clip}\!\left(\!\rho\Im\{\mathbf{x}\}, \xi \right),
\label{eq:proxg}
\end{align}
where the clipping function $\text{clip}(\mathbf{x}, \gamma)$ applies the operation $\min(\max(x_i, -\gamma), \gamma)$ to each element of the vector $\mathbf{x}$.

We observe that there are two parameters in the C2PO algorithm, namely, $\tau$ and $\rho$, where $\tau$ can be seen as the step size and $\rho$ corresponds to the clipping boundary. The C2PO algorithm converges to a stationary point provided that the step size is chosen appropriately. One sufficient condition of this convergence is that $\tau < \|\mathbf{A}^H\mathbf{A}\|_{2,2}^{-1}$, and $\tau \delta < 1$ \cite{castaneda_1-bit_2017}.

Unfortunately, the ideal selection of the parameters, $\tau$ and $\rho$ depend on the system configuration, the modulation scheme and the channel. To achieve excellent performance, these parameters must be tuned to ensure a fast convergence \cite{chen_alternating_2018}. The work of \cite{castaneda_1-bit_2017} suggested some choices for these parameters, which are empirical values and are constant for different C2PO iterations. On the other hand, \cite{balatsoukas-stimming_neural-network_2019} proposed a learning-based optimization tool (NNO-C2PO) to tune these parameters, which outperforms the empirical values in~\cite{castaneda_1-bit_2017}. The idea of NNO-C2PO is to 'unfold' the C2PO algorithm and to learn its parameters through backpropagation from training samples. 

\begin{figure*}
	\centering
	\includegraphics[width=0.975\textwidth]{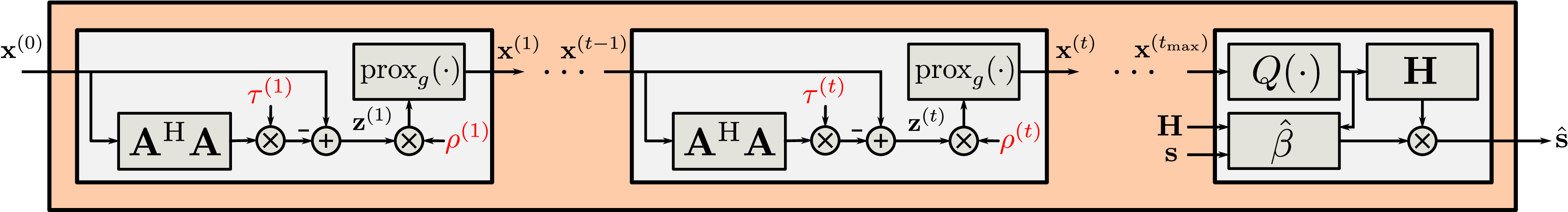}
	\caption{Computation graph of the iteration-unfolded version of NNO-C2PO. All trainable parameters are highlighted in red color \cite{balatsoukas-stimming_neural-network_2019}.}\label{fig:nnofap}
\end{figure*}

The work of \cite{balatsoukas-stimming_neural-network_2019} first determined the number of C2PO iterations and unfolded (\ref{eq:c2postep1}) and (\ref{eq:c2postep2}), with layer-specific $\tau$ and $\rho$ as trainable parameters, which are denoted as $\tau^{(t)}$ and $\rho^{(t)}$. As shown in Fig. \ref{fig:nnofap}, the input of the C2PO algorithm is  $\mathbf{x}^{(0)} = \mathbf{H}^H \mathbf{s}$, and each following  blocks corresponds to a C2PO iteration. After $t_\mathrm{max}$ iterations, the output is fed into a 1-bit quantizer, then transmitted through the channel $\mathbf{H}$ and scaled with an $\mathbf{H}$-  and $\mathbf{s}$-dependent factor $\hat{\beta}$ to yield the estimated received vector $\hat{\mathbf{s}}$. By setting the loss function as the mean-squared-error (MSE) loss between ground truth symbol and the estimated symbol
\begin{align}
C = \|\mathbf{s} - \hat{\mathbf{s}}\|_2^2, 
\label{eq:cost}
\end{align}
the MSE-optimal $\tau^{(t)}$ and $\rho^{(t)}$ can be learned from backpropagation while minimizing $C$.

\section{Channel Selection for Parameter Learning}
\label{sec:channel_selection}
%In the following, we consider a scenario with with $U = 8$ users, $B = 128$ base station antennas and 16-QAM modulation. 
%We start from a training set formed by random Rayleigh channels $\mathbf{H}$ and symbols $\mathbf{s}$ in a scenario 
%C2PO parameters $\tau^{(t)}$ and $\rho^{(t)}$ are learned using the structure shown in Fig. \ref{fig:nnofap} using the NNO-C2PO algorithm proposed in \cite{balatsoukas-stimming_neural-network_2019}. 
For unfolded parameter learning, a training set should be generated first which is then used for training the tunable parameters in the network as shown in Fig. \ref{fig:nnofap}. For testing, the C2PO algorithm is then run with those learned parameters on randomly generated channels and symbols to find the corresponding precoded symbols $\mathbf{\hat{x}}$. These precoded symbols are then sent through the channels, additive white Gaussian noise (AWGN) is added, and the performance is evaluated.

\subsection{Metric for Evaluating the Performance of C2PO}
There are different methods to evaluate the performance of a precoder by running a large number of simulations. A straightforward metric is to measure the average squared distance between desired received vectors and the actual received vectors, similar to the MSE cost function for the network as shown in (\ref{eq:cost}). However, the performance of a wireless system is ultimately better characterized by the symbol-error rate (SER) as a function of the SNR. We notice that some SER curves in the work of \cite{balatsoukas-stimming_neural-network_2019} stop improving for higher SNR. It is obvious that this part of the curve is dominated by the received symbols that have such a large distance from their original location that they cross the decision boundary to another symbol even in the absence of AWGN. The resulting error ultimately limits the performance of the wireless system. Hence, we decide to focus on this \emph{error floor} as our main performance criterion.

\subsection{Impact of C2PO Parameters on Convergence}
\label{sec:convergence}
As mentioned in \cite{balatsoukas-stimming_neural-network_2019}, the performance of C2PO 
strongly depends on the parameters $\tau^{(t)}$ and $\rho^{(t)}$, which ideally should be optimized for each individual channel instance. However, learning these parameters for every new channel instance is prohibitively complex. Hence, we are interested in finding a training set that learns $\tau^{(t)}$ and $\rho^{(t)}$ such that the C2PO algorithm achieves a low average error floor across a population of channels. The most obvious choice is to train the algorithm with channel samples that are drawn from the same distribution as the test set during operation, as done in~\cite{balatsoukas-stimming_neural-network_2019}. Unfortunately, the MSE quality metric employed during training in \cite{balatsoukas-stimming_neural-network_2019} is only a rough proxy for both the error floor and the error rate performance. For the error floor metric, the error is dominated by some non-converging samples, whereas the MSE metric averages out the loss contributed by those challenging cases. Therefore, the training procedure with the MSE metric may neglect the influence of those non-converging samples. However, training with an error floor cost function is difficult since the error rate function and, hence, also the error floor function is non-differentiable.

We therefore propose to carefully choose training sets that lead to C2PO parameters that are unlikely to cause poor convergence for the majority of the channels to avoid (or at least reduce) error floors. To this end, it is useful to first consider the case of a single parameter 
$\tau^{(t)}=\tau$ and its impact on the C2PO iterative algorithm. This insight will help us to i) understand how this parameter is chosen during training and ii) how, once it is fixed, it affects performance. In fact, the step size $\tau$ has a profound effect on the convergence of C2PO, as shown in the following theorem~\cite{castaneda_1-bit_2017}.

\begin{theorem}
Suppose the step size used in C2PO satisfies $\tau < \|\mathbf{A}^H \mathbf{A} \|_{2,2}^{-1}$, and $\tau\delta < 1$. Then, C2PO decreases the objective (\ref{equ:obj_C2PO}) monotonically, and any limit point of the iterates $\{\mathbf{x}^{(t)}\}$ is a stationary point.
\label{theo:converge_AHA}
\end{theorem}

Hence, for $\tau < \|\mathbf{A}^H \mathbf{A}\|_{2,2}^{-1}$, each iteration of C2PO will keep the objective (\ref{equ:obj_C2PO}) converging towards a local minimum. 
Unfortunately, once $\tau$ is fixed after training, we may still encounter combinations of channels and symbol vectors for which $\tau > \|\mathbf{A}^H \mathbf{A}\|_{2,2}^{-1}$ so that the C2PO algorithm diverges, which contributes significantly to the error floor. An intuitive solution is to adapt $\tau$ according to $\mathbf{A}$. However, since $\mathbf{A}$ is a function of $\mathbf{H}$, which may be changing rapidly, determining $\tau$ for each $\mathbf{A}$ has a very high computational cost. To avoid this issue, we can simply choose a small $\tau$ to ensure that the vast majority of the matrices $\mathbf{A}$ meet the convergence condition. Unfortunately, it is also not useful to set $\tau$ too small since $\tau$ acts as the step size in the C2PO gradient descent update step \eqref{eq:c2postep1}. Hence, with a limited number of iterations, $\tau$ still needs to be sufficiently large to substantially reduce the objective for a sufficient number of channels. We therefore expect that the best choice for $\tau$ is neither a very large nor a very small value and main challenge lies in determining such a good value. We note that the parameter $\rho = \frac{1}{1-\tau\delta}$ will be learned as well once $\tau$ is learned properly.

\subsection{Training C2PO for a Population of Channels}
\label{sec:traing_C2PO}
With only a scalar parameter $\tau$ and the insight from above, this parameter could in principle be optimized with a simple grid-like search. However, the NNO-C2PO algorithm has many per-iteration parameters $\tau^{(t)}$ and $\rho^{(t)}$, which are optimized individually by using the training process and are essential to improve the performance of the original C2PO algorithm. This makes a  grid-like search impractical. 
Instead, in order to achieve a set of parameters that meet the convergence condition of C2PO for most channels and still also avoid too slow convergence, we propose to adjust the choice of training examples based on the insight from Section~\ref{sec:convergence}. By recalling Theorem~1, it is intuitive that the distribution of the 2-norms $\|\mathbf{A}^H \mathbf{A}\|_{2,2}$ of the matrices $\mathbf{A}$ in the training set influences the magnitude of the resulting $\tau^{(t)}$. More specifically, skewing the distribution of training examples toward larger/smaller 2-norms decreases/increases the resulting values of~$\tau^{(t)}$.

We therefore propose to generate training sets with examples of $\mathbf{A}$ with a specific and adjustable 2-norm, while still maintaining a sufficiently rich pool of training examples that structurally resemble the channels encountered during operation. Unfortunately, it is difficult to generate $\mathbf{A}$ with a certain 2-norm, since $\mathbf{A}$ is a function of both $\mathbf{s}$ and $\mathbf{H}$. Recalling that $\mathbf{A} = (\mathbf{I}_U - \mathbf{s}\mathbf{s}^H/\|\mathbf{s}\|_2^2)\mathbf{H}$, we neglect the influence of $\mathbf{s}$ since the 2-norm of $\mathbf{H}$, provides also a sufficient condition for the convergence of C2PO, as shown in Proposition~\ref{prop:converge_HHH}.

\begin{proposition}
Suppose the step size used in C2PO satisfies $\tau < \|\mathbf{H}^H \mathbf{H} \|_{2,2}^{-1}$, and $\tau\delta < 1$. Then, the condition in Theorem~\ref{theo:converge_AHA} is fulfilled.
\label{prop:converge_HHH}
\end{proposition}
\begin{proof}
As $\mathbf{A} = (\mathbf{I}_U - \mathbf{s}\mathbf{s}^H/\|\mathbf{s}\|_2^2)\mathbf{H}$, where $(\mathbf{I}_U - \mathbf{s}\mathbf{s}^H/\|\mathbf{s}\|_2^2)$ is a projection matrix and always has a unity 2-norm, we have
\begin{align*} 
\|\mathbf{A}\|_{2,2} 
&= \|(\mathbf{I}_U - \mathbf{s}\mathbf{s}^H/\|\mathbf{s}\|_2^2)\mathbf{H}\|_{2,2}\\
&\leq \|(\mathbf{I}_U - \mathbf{s}\mathbf{s}^H/\|\mathbf{s}\|_2^2)\|_{2,2} \|\mathbf{H}\|_{2,2}\\
%&= 1 \cdot \|\mathbf{H}\|_{2,2}\\
&= \|\mathbf{H}\|_{2,2},
\end{align*}
and thus
\begin{align}
\tau < \|\mathbf{H}^H \mathbf{H} \|_{2,2}^{-1} \leq \|\mathbf{A}^H \mathbf{A} \|_{2,2}^{-1},
\label{eq:prop}
\end{align}
so that the condition of Theorem~\ref{theo:converge_AHA} is met.
\end{proof}

\subsection{2-norm-based channel set}
For our proposed training strategy, we need to achieve a rich pool of channels with a similar given 2-norm. One straightforward method is to generate a large number of random channels from the target channel model, and then group these channels based on their 2-norm. The work of \cite{hansen_2-norm_1988} argued that the 2-norm of a random matrix follows a normal distribution according to the size of the matrix and the variance of each entry. As an example, an $8\times128$ channel matrix whose entries $h_{i, j} \sim \mathcal{N}(0, 1)$, has the 2-norm $\|\mathbf{H}\|_{2,2} \sim \mathcal{N}(13.5, 0.5^2)$. Therefore, it is very time consuming to generate many channels with large 2-norms. Thus, we decide to generate artificial channels with the desired 2-norms for training. We propose two algorithms to transform randomly generated channels based on the given channel model into training channels with the desired 2-norm $\|\mathbf{H}\|_{2,2} = N$. In our experiments, we use an i.i.d. Rayleigh fading channel model. 

The first algorithm is to simply scale the channels according to their 2-norm values with the following process.

\begin{oframed}

\vspace{-0.3cm}

\begin{alg} []\label{alg:gen_H1} 

Generate a random matrix $\mathbf{H'}$ using the desired channel model and compute the singular value decomposition of $\mathbf{H'}$: 
\begin{align}
\mathbf{H'} = \mathbf{U}\mathbf{\Lambda'}\mathbf{V}^H
\end{align}
Scale the other singular values 
\begin{align}
\mathbf{\Lambda} = \gamma \mathbf{\Lambda'}
\end{align}
such that 
$\lambda_{\text{max}} = N$ and construct a training example $\mathbf{H}$ by 
\begin{align}
\mathbf{H} = \mathbf{U}\mathbf{\Lambda}\mathbf{V}^H.
\end{align}
\end{alg}

\vspace{-0.2cm}
\end{oframed}

However, by scaling the singular values of the channel matrix, the Frobenius norm (F-norm) of the channel matrix is also changed. In order to exclude the influence of scaling the F-norm from the training set, we propose Algorithm \ref{alg:gen_H2}. In this algorithm, we manipulate the singular values separately to change the 2-norm, while keeping the F-norm fixed. This method has two drawbacks: First, it still involves Monte Carlo trials for each example and channels with very large 2-norms are still rare since the F-norm is limited by that of the randomly chosen starting point. The second drawback is that since the singular values are manipulated separately, there might be the case that the second largest singular value, after manipulation, becomes the largest singular value, which in turn would, by definition, change the 2-norm. To deal with the second problem, we simply drop all generated channel matrices whose 2-norm is not equal to $N$ after the transformation.

\begin{oframed}

\vspace{-0.3cm}

\begin{alg} []\label{alg:gen_H2} 
Start by finding a good baseline matrix $\mathbf{H'}$
\Repeat{$\lambda_2 \leq N$}{
Generate a random matrix $\mathbf{H'}$ using the desired channel model and compute the singular value decomposition of $\mathbf{H'}$: 
\begin{align}
\mathbf{H'} = \mathbf{U}\mathbf{\Lambda'}\mathbf{V}^H,
\end{align}
where the singular values are in descending order.

\vspace{0.15cm}
Set the largest singular value $\lambda_1 = N$.
Scale the other singular values 
\begin{align}
\lambda_i = \gamma\lambda'_i, \quad \text{where} \quad i = 2,\dots, U,
\end{align}
such that 
$\sum_{i = 1}^U \lambda^2_i = \sum_{i = 1}^U \lambda'^{2}_i$ 
}

\vspace{0.15cm}
Construct $\mathbf{H}$ by 
\begin{align}
\mathbf{H} = \mathbf{U}\mathbf{\Lambda}\mathbf{V}^H
\end{align}
\end{alg}
\vspace{-0.2cm}
\end{oframed}

For both algorithms, we expect decreasing values of learned $\tau^{(t)}$ with increasing $N$ in the training set. This decreasing trend of $\tau^{(t)}$ will first reduce the error floor of C2PO since more channels meet the convergence condition, but eventually again increase the error floor because the decreasing step size limits the C2PO updates with a limited number of C2PO updates.

\section{Results}
\label{sec:results}
To show that the selection of the training samples improves the performance of the learned NNO-C2PO, we now compare the error floor after training with channel matrices with different 2-norms $N$. The considered training alternatives are 1) a training set of randomly generated channels which follow the Rayleigh distribution as in~\cite{balatsoukas-stimming_neural-network_2019}, denoted as the default set (\textit{DF-set}); 2) a training set generated with 2-norm $N$ selected according to the median value and from the tails of the distribution of the DF-set, denoted as norm-tuned-sets (\textit{NT-set1} generated with Algorithm \ref{alg:gen_H1}; 3) a training set \textit{NT-set2} generated from Algorithm~\ref{alg:gen_H2}), also with different target 2-norms. For training, we implement C2PO with various iterations and unfold the algorithms as described in~\cite{balatsoukas-stimming_neural-network_2019}. The parameters are trained on the training sets in a scenario with $U = 8$, $B = 128$, and with 16-QAM modulation. We generate $K = 500$ training samples and perform full batch training until the cost is stable. Similar to~\cite{balatsoukas-stimming_neural-network_2019}, the computation graph is implemented using Keras \cite{chollet_keras_2015} with a TensorFlow \cite{tensorflow2015-whitepaper} backend. After the parameters $\tau^{(t)}$ and $\rho^{(t)}$ have been learned, we simulate the learned C2PO algorithm on the randomly generated Rayleigh channels whose distribution is the same as that of the DF-set to obtain both the error floor and the SER.

\subsection{Impact of training set on learned C2PO}
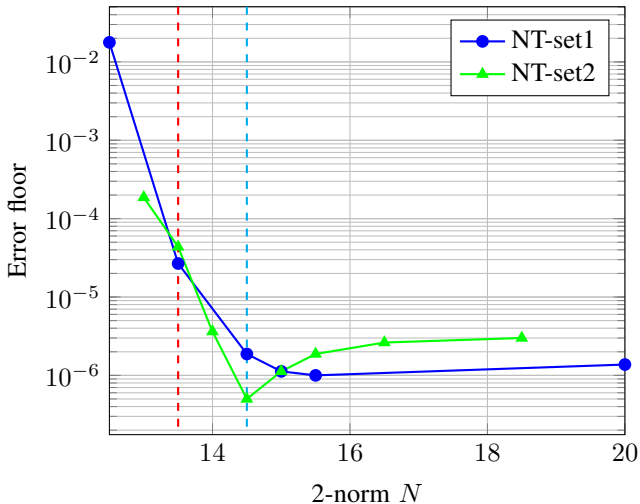
\begin{figure}
\begin{tikzpicture}
%%%%%%%%IMPORTANT%%%%%%%%%%%%%%%%%%%%%%%%%%%%%   
% V1 WITHOUT the simulation point N = 30
\begin{axis}
[legend cell align=left, legend pos = north east, ymode=log,xlabel={2-norm $N$}, ylabel=Error floor, grid=both, xmin=12.5,xmax=20]%, xtick= {12.5, 13.5, 14.5, 15.5, 16.5, 18.5, 20}]
	\addplot [thick, color=blue,
    mark=oplus*]
    table [x=normT, y=error_floor, col sep=comma] 		{./csv/trainNorm_floor_new_V1.txt};\addlegendentry{NT-set1} 
	\addplot [thick, color=green,
    mark= triangle*]
    table [x=new_normT, y=new_error_floor, col sep=comma] {./csv/trainNorm_floor_new_V1.txt}; \addlegendentry{NT-set2}
\draw [thick, color = red, dashed]({axis cs:13.5,0}|-{rel axis cs:0,1}) -- ({axis cs:13.5,0}|-{rel axis cs:0,0});
\draw [thick, color = cyan, dashed]({axis cs:14.5,0}|-{rel axis cs:0,1}) -- ({axis cs:14.5,0}|-{rel axis cs:0,0});
\end{axis}
\end{tikzpicture}
\caption{Performance on error floor of C2PO ($t_{\text{max}} = 7$) learned from training set with increasing channel matrix 2-norm. The cyan dashed line denotes $N = 14.5$ and the red dashed line denotes $N = 13.5$, which is the same as the median value of 2-norm in DF-set. The simulation is based on the scenario of $U = 8$, $B = 128$ and 16-QAM modulation.
}
\label{fig:trainNorm_floor}
\end{figure} 
We learn C2PO parameters from NT-sets and test the learned C2PO on the randomly generated Rayleigh channels. The error floor for $U=8$, $B=128$, and $t_{\max} = 7$ is ploted in Fig.~\ref{fig:trainNorm_floor}. As expected, for both NT-sets, the error floor first decreases as the 2-norm of channels in the training set increases, since more channels in the test set converge with the trained $\tau^{(t)}$. However, as the 2-norm of channels in the training set further increases, the error floor increases again since the learned step size $\tau^{(t)}$ becomes too small for most of the channels in the test set. 
Moreover, we can clearly observe that the 2-norm that is obtained by simply choosing the entries of $\mathbf{H}$ as i.i.d. standard complex normal variables, which is denoted with the red dashed line in Fig.~\ref{fig:trainNorm_floor}, results in an almost two orders of magnitude higher error floor than with a more carefully chosen training set with a 2-norm $N=14.5$. 
Finally, we observe that NT-set1 and NT-set2, generated using Algorithm~\ref{alg:gen_H1} and Algorithm~\ref{alg:gen_H2}, respectively, have a similar error floor performance, which shows that the 2-norm of the training channels has a more important impact on the training than the F-norm. Nevertheless, Algorithm~\ref{alg:gen_H1} seems to be less sensitive to the exact choice of $N$, while Algorithm~\ref{alg:gen_H2} results in a slightly lower error floor for the best value of $N$.

\subsection{Error floor for learned C2PO}
We now focus on three specific training sets and their performance, namely, the DF-set which follows the same (Rayleigh) distribution as the test set and has the median 2-norm value of $N=13.5$, the NT-set1 whose 2-norm is $N=14.5$ generated from Algorithm~\ref{alg:gen_H1}, and the NT-set2 generated from Algorithm~\ref{alg:gen_H2} whose 2-norm is also $N=14.5$. We evaluate the error floor for different number of C2PO iterations and also show the performance in terms of the symbol error rate for different SNRs.
We compare these three training generation schemes using the error floor metric described in Section~\ref{sec:channel_selection}. The C2PO iteration number is tested for $t_\text{max} = \{2,4,6,7\}$. As shown in Fig.~\ref{fig:perform_on_step}, the C2PO with coefficients trained with the DF-set has a larger error floor compared to the C2PO trained on the NT-sets with $N = 14.5$, even though the DF-set shares the same structural and 2-norm distribution with test set.

\begin{figure}[t]
\centering
\begin{tikzpicture}
%%%%%%%%%%IMPORTANT%%%%%%%%%%%%%%%%%%%%%%%%%
% V1 WITHOUT point iteration 3,5
\begin{axis}[legend cell align=left, ymode=log,xlabel={Number of C2PO iterations}, ylabel=Error floor, grid=both, xtick={2,4,6,7}, xmin=2, xmax=7, legend pos=south west, ymin=1e-7]
\addplot [thick, color=red,
    mark=square*, only marks ]
    table [x=step, y=all_on_all, col sep=comma] {./csv/step_error_floor_V1.txt};\addlegendentry{DF-set}
\addplot [thick, color=blue,
    mark=oplus*, only marks ]
    table [x=step, y=allmaxold_on_all, col sep=comma] {./csv/step_error_floor_V1.txt};\addlegendentry{NT-set1 ($N = 14.5$)}
\addplot [thick, color=green,
    mark=triangle*, only marks ]
    table [x=step, y=allmaxnew_on_all, col sep=comma] {./csv/step_error_floor_V1.txt};\addlegendentry{NT-set2 ($N = 14.5$)}
\end{axis}
\end{tikzpicture}
\vspace{-0.3cm}
 \caption{Error floor of C2PO with $t_\text{max} = \{2,4,6,7\}$ trained from DF-set and NT-sets with $N = 14.5$. The simulation is based on the scenario of $U = 8$, $B = 128$ and 16-QAM modulation.} 
\label{fig:perform_on_step}
\end{figure}
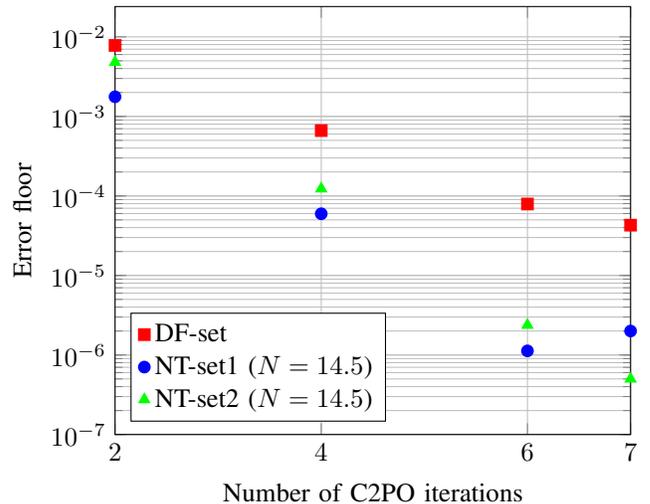

\subsection{SER for learned C2PO}
By fixing the number of iterations $t_\text{max} = 7$, we test the performance of the DF-set and the NT-sets with $N = 14.5$ on C2PO for different SNRs. As shown in Fig.~\ref{fig:SER-SNR}, for low SNR, the C2PO learned from NT-sets maintains a similar performance as the one learned from the DF-set. For high SNR, however, the C2PO learned from the NT-sets reduces the SER level significantly compared to the C2PO trained on the DF-set.

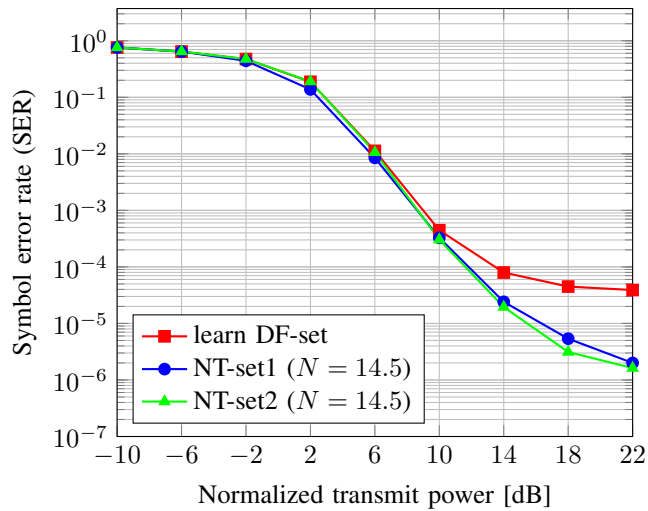
\begin{figure}
\begin{tikzpicture}
\begin{axis}[legend cell align=left, ymode=log,xlabel={Normalized transmit power [dB]}, ylabel=Symbol error rate (SER), grid=both,ymin=1e-7, xtick={-10,-6,-2,2,6,10,14,18,22}, xmin=-10,xmax=22,legend pos=south west,]
\addplot [thick, color=red,
    mark=square*, ]
    table [x=SNR, y=all_on_all, col sep=comma] {./csv/SER-SNR.txt};\addlegendentry{learn DF-set}
\addplot [thick, color=blue,
    mark=oplus*, ]
    table [x=SNR, y=allmax_old_on_all, col sep=comma] {./csv/SER-SNR.txt};
\addlegendentry{NT-set1 ($N = 14.5$)}
\addplot [thick, color=green,
    mark=triangle*, ]
    table [x=SNR, y=allmax_new_on_all, col sep=comma] {./csv/SER-SNR.txt};
\addlegendentry{NT-set2 ($N = 14.5$)}
\end{axis}
\end{tikzpicture}
\caption{Symbol error rate (SER) of C2PO ($t_\text{max} = 7$) trained from DF-set and NT-sets for different normalized transmit power. The simulation is based on the scenario of $U = 8$, $B = 128$ and 16-QAM modulation.}
\label{fig:SER-SNR}
\end{figure}

\section{Conclusion}
\label{sec:conclusion}
Using the error floor as a performance metric, we have shown how the tunable parameters in C2PO influence the convergence of C2PO and the error rate performance of the system. Instead of generating training samples which are distributed as those in the test set, we propose a training channel generation scheme which generates training sets based on the 2-norm of the channel matrix. Simulations show that by selecting the training channel properly, the error floor of the symbol error rate is significantly reduced compared to the case where the training channel set follows the same distribution as the test set.

\section*{Acknowledgments}
This research has been kindly supported by the Swiss National Science Foundation under Grant-ID 182621.

\balance
\bstctlcite{IEEEexample:BSTcontrol}
\bibliographystyle{IEEEtran}  
\bibliography{./reference}

\end{document}